\documentclass{article}%
\usepackage{amssymb}
\usepackage{amsfonts}
\usepackage{amsmath}
\usepackage{graphicx}%
\providecommand{\U}[1]{\protect\rule{.1in}{.1in}}
\newtheorem{theorem}{Theorem}

\newtheorem{lemma}[theorem]{Lemma}

\newenvironment{proof}[1][Proof]{\noindent\textbf{#1.} }{\ \rule{0.5em}{0.5em}}
\begin{document}

\title{Effect of the Interconnected Network Structure on the Epidemic Threshold}
\author{Huijuan Wang$^{1,2}$, Qian Li$^{2}$, Gregorio D'Agostino$^{3}$,
\and Shlomo Havlin$^{4}$ and H. Eugene Stanley$^{2}$ and Piet Van Mieghem$^{1}$\\$^{1}$Faculty of Electrical Engineering, Mathematics, and Computer Science, \\Delft University of Technology, Delft, The Netherlands\\$^{2}$Center for Polymer Studies and Department of Physics, \\Boston University, Boston, MA 02215 USA\\$^{3}$ENEA - CR "Casaccia," via Anguillarese 301, I-00123 Roma, Italy\\$^{4}$Department of Physics, Bar-Ilan University, 52900 Ramat-Gan, Israel}
\maketitle

\begin{abstract}
Most real-world networks are not isolated. In order to function fully, they
are interconnected with other networks, and this interconnection influences
their dynamic processes. For example, when the spread of a disease involves
two species, the dynamics of the spread within each species (the contact
network) differs from that of the spread between the two species (the
interconnected network). We model two generic interconnected networks using
two adjacency matrices, A and B, in which A is a $2N\times2N$ matrix that
depicts the connectivity within each of two networks of size $N$, and B a
$2N\times2N$ matrix that depicts the interconnections between the two. Using
an N-intertwined mean-field approximation, we determine that a critical
susceptable-infected-susceptable (SIS) epidemic threshold in two
interconnected networks is $1/\lambda_{1}(A+\alpha B)$, where the infection
rate is $\beta$ within each of the two individual networks and $\alpha\beta$
in the interconnected links between the two networks and $\lambda_{1}(A+\alpha
B)$ is the largest eigenvalue of the matrix $A+\alpha B$. In order to
determine how the epidemic threshold is dependent upon the structure of
interconnected networks, we analytically derive $\lambda_{1}(A+\alpha B)$
using perturbation approximation for small and large $\alpha$, the lower and
upper bound for any $\alpha$ as a function of the adjacency matrix of the two
individual networks, and the interconnections between the two and their
largest eigenvalues/eigenvectors. We verify these approximation and boundary
values for $\lambda_{1}(A+\alpha B)$ using numerical simulations, and
determine how component network features affect $\lambda_{1}(A+\alpha B)$. We
note that, given two isolated networks $G_{1}$ and $G_{2}$ with principle
eigenvectors $x$ and $y$ respectively, $\lambda_{1}(A+\alpha B)$ tends to be
higher when nodes $i$ and $j$ with a higher eigenvector component product
$x_{i}y_{j}$ are interconnected. This finding suggests essential insights into
ways of designing interconnected networks to be robust against epidemics.

\end{abstract}

\section{Introduction}

Complex network studies have traditionally focused on single networks in which
nodes represent agents and links represent the connections between agents.
Recent efforts have focused on complex systems that are comprised of
interconnected networks, a configuration that more accurately represents
real-world networks \cite{Buldyrev2010,Raissa}. Real-world power grids, for
example, are almost always coupled with communication networks. Power stations
need communication nodes for control and communication nodes need power
stations for electricity. The influence of coupled networks on cascading
failures has been widely studied
\cite{Buldyrev2010,ParshaniPRL2010,Gao,Zhou_inter,Huang_inter}. When a node at
one end of an interdependent link fails, the node at the other end of the link
usually fails. A non-consensus opinion model of two interconnected networks
that allows the opinion interaction rules within each individual network to
differ from those between the networks was recently studied \cite{Li}. This
model shows that opinion interactions between networks can transform
non-consensus opinion behavior into consensus opinion behavior.

In this paper we investigate the susceptable-infected-susceptable (SIS)
behavior of a spreading virus, a dynamic process in interconnected networks
that has received significant recent attention
\cite{Funk,Dickison,Mendiola,Sahneh}. An interconnected networks scenario is
essential when modeling epidemics because diseases spread across multiple
networks, e.g., across multiple species or communities, through both contact
network links within each species or community and interconnected network
links between them. Dickison et al. \cite{Dickison} study the behavior of
susceptible-infected-recovered (SIR) epidemics in interconnected networks.
Depending on the infection rate in weakly and strongly coupled network
systems, where each individual network follows the configuration model and
interconnections are randomly placed, epidemics will infect none, one, or both
networks of a two-network system. Mendiola et al. \cite{Mendiola} show that in
SIS model an endemic state may appear in the coupled networks even when an
epidemic is unable to propagate in each network separately. In this work we
will explore how both the structural properties of each individual network and
the behavior of the interconnections between them determine the epidemic
threshold of two generic interconnected networks.

In order to represent two generic interconnected networks, we represent a
network $G$ with $N$ nodes using an $N\times N$ adjacency matrix $A_{1}$ that
consists of elements $a_{ij}$, which are either one or zero depending on
whether there is a link between nodes $i$ and $j$. For the interconnected
networks, we consider two individual networks $G_{1}$ and $G_{2}$ of the same
size $N$. When nodes in $G_{1}$ are labeled from $1$ to $N$ and in $G_{2}$
labeled from $N+1$ to $2N$, the two isolated networks $G_{1}$ and $G_{2}$ can
be presented by a $2N\times2N$ matrix $A=\left[
\begin{array}
[c]{cc}%
A_{1} & \mathbf{0}\\
\mathbf{0} & A_{2}%
\end{array}
\right]  $ composed of their corresponding adjacency matrix $A_{1}$ and
$A_{2}$ respectively. Similarly, a $2N\times2N$ matrix $B=\left[
\begin{array}
[c]{cc}%
\mathbf{0} & B_{12}\\
B_{12}^{T} & \mathbf{0}%
\end{array}
\right]  $ represents the symmetric interconnections between $G_{1}$ and
$G_{2}$.\ The interconnected networks are composed of three network
components: network $A_{1}$, network $A_{2}$, and interconnecting network $B$.

In the SIS model, the state of each agent at time $t$ is a Bernoulli random
variable, where $X_{i}(t)=0$ if node $i$ is susceptible and $X_{i}(t)=1$ if it
is infected. The recovery (curing) process of each infected node is an
independent Poisson process with a recovery rate $\delta$. Each infected agent
infects each of its susceptible neighbors with a rate $\beta$, which is also
an independent Poisson process. The ratio $\tau\triangleq\beta/\delta$ is the
effective infection rate. A phase transition has been observed around a
critical point $\tau_{c}$ in a single network. When $\tau>\tau_{c}$, a
non-zero fraction of agents will be infected in the steady state, whereas if
$\tau<\tau_{c}$, infection rapidly disappears \cite{May,Satorras2012}. The
epidemic threshold via the N-intertwined mean-field approximation (NIMFA) is
$\tau_{c}=\frac{1}{\lambda_{1}(A)}$, where $\lambda_{1}(A)$ is the largest
eigenvalue of the adjacency matrix, also called the spectral radius
\cite{epiton}. For interconnected networks, we assume that the curing rate
$\delta$ is the same for all the nodes, that the infection rate along each
link of $G_{1}$ and $G_{2}$ is $\beta$, and that the infection rate along each
interconnecting link between $G_{1}$ and $G_{2}$ is $\alpha\beta$, where
$\alpha$ is a real constant ranging within $[0,\infty)$ without losing generality.

We first show that the epidemic threshold for $\beta/\delta$ in interconnected
networks via NIMFA is $\tau_{c}=\frac{1}{\lambda_{1}(A+\alpha B)}$, where
$\lambda_{1}(A+\alpha B)$ is the largest eigenvalue of the matrix $A+\alpha
B$. We further express $\lambda_{1}(A+\alpha B)$ as a function of each network
component $A_{1}$, $A_{2}$, and $B$ and their eigenvalues/eigenvectors to
reveal the contribution of each component network. This is a significant
mathematical challenge, except for special cases, e.g., when $A$ and $B$
commute, i.e., $AB=BA$ (see Sec.~\ref{special cases}). Our main contribution
is that we analytically derive for the epidemic characterizer $\lambda
_{1}(A+\alpha B)$ (a) its perturbation approximation for small $\alpha$, (b)
its perturbation approximation for large $\alpha$, and (c) its lower and upper
bound for any $\alpha$ as a function of component network $A_{1}$, $A_{2}$,
and $B$ and their the largest eigenvalues/eigenvectors. Numerical simulations
in Sec.~\ref{simulation} verify that these approximations and bounds well
approximate $\lambda_{1}(A+\alpha B)$, and thus reveal the effect of component
network features on the epidemic threshold of the whole system of
interconnected networks, which provides essential insights into designing
interconnected networks that are robust against the spread of epidemics (see
Sec.~\ref{conclusion}).

Sahneh et al. \cite{Sahneh} recently studied SIS epidemics on generic
interconnected networks in which the infection rate can differ between $G_{1}$
and $G_{2}$, and derived the epidemic threshold for the infection rate in one
network while assuming that the infection does not survive in the other. Their
epidemic threshold was expressed as the largest eigenvalue of a function of
matrices. Our work explains how the epidemic threshold of generic
interconnected networks is related to the properties (eigenvalue/eigenvector)
of each network component $A_{1}$, $A_{2}$, and $B$ without any approximation
on the network topology.

Graph spectra theory \cite{PVM_Spectrabook} and modern network theory,
integrated with dynamic systems theory, can be used to understand how network
topology can predict these dynamic processes. Youssef and Scoglio \cite{sir}
have shown that a SIR epidemic threshold via NIMFA also equals $1/\lambda_{1}%
$. The Kuramoto synchronization process of coupled oscillators \cite{Restrepo}
and percolation \cite{percolation} also features a phase transition that
specifies the onset of a remaining fraction of locked oscillators and the
appearance of a giant component, respectively. Note that a mean-field
approximation predicts both phase transitions at a critical point that is
proportional to $1/\lambda_{1}$. Thus we expect our results to apply to a
wider range of dynamic processes in interconnected networks.

\section{Epidemic Threshold of Interconnected Networks}

In the SIS epidemic spreading process, the probability of infection
$v_{i}(t)=E[X_{i}(t)]$ for a node $i$ in interconnected networks $G$ is
described by
\[
\frac{dv_{i}(t)}{dt}=\left(  \beta\sum_{j=1}^{2N}a_{ij}v_{j}(t)+\alpha
\beta\sum_{j=1}^{2N}b_{ij}v_{j}(t)\right)  \left(  1-v_{j}(t)\right)  -\delta
v_{j}(t)
\]
via NIMFA, where $a_{ij}$ and $b_{ij}$ is an element of matrix $A$ and $B$
respectively. Its matrix form becomes
\[
\frac{dV(t)}{dt}=\left(  \beta\left(  A+\alpha B\right)  V(t)-\delta I\right)
-\beta diag\left(  v_{i}(t)\right)  \left(  A+\alpha B\right)  V(t).
\]
The governing equation of the SIS spreading process on a single network
$A_{1}$ is
\[
\frac{dV(t)}{dt}=\left(  \beta A_{1}V(t)-\delta I\right)  -\beta diag\left(
v_{i}(t)\right)  A_{1}V(t),
\]
whose epidemic threshold has been proven \cite{epiton} to be
\[
\tau_{c}=\frac{1}{\lambda_{1}(A_{1})},
\]
which is a lower bound of the epidemic threshold \cite{cator}. Hence, the
epidemic threshold of interconnected networks by NIMFA is
\begin{equation}
\tau_{c}=\frac{1}{\lambda_{1}(A+\alpha B)} \label{epidemicthreshold}%
\end{equation}
which depends on the largest eigenvalue of the matrix $A+\alpha B$. The matrix
$A+\alpha B$ is a weighted matrix, where $0\leq\alpha<\infty$. Note that the
NIMFA model is an improvement over earlier epidemic models \cite{Satorras2012}
in that it applies no approximations to network topologies, and thus it allows
us to identify the specific role of a general network structure on the
spreading process.

\section{Analytic approach: $\lambda_{1}(A+\alpha B)$ in relation to component
network properties}

\label{analysis}

The spectral radius $\lambda_{1}(A+\alpha B)$ as shown in the last section is
able to characterize epidemic spreading in interconnected networks. In this
section we explore how $\lambda_{1}(A+\alpha B)$ is influenced by the
structural properties of interconnected networks and by the relative infection
rate $\alpha$ along the interconnection links. Specifically, we express
$\lambda_{1}(A+\alpha B)$ as a function of the component network $A_{1}$,
$A_{2}$, and $B$ and their eigenvalues/eigenvectors. (For proofs of theorems
or lemma, see the Appendix.)

\subsection{Special cases}

\label{special cases}

We start with some basic properties related to $\lambda_{1}(A+\alpha B)$ and
examine several special cases in which the relation between $\lambda
_{1}(A+\alpha B)$ and the structural properties of network components $A_{1}$,
$A_{2}$ and $B$ are analytically tractable.

The spectral radius of a sub-network is always smaller or equal to that of the
whole network. Hence,

\begin{lemma}%
\[
\lambda_{1}(A+\alpha B)\geq\lambda_{1}(A)=\max\left(  \lambda_{1}
(A_{1})\text{, }\lambda_{1}(A_{2})\right)
\]

\end{lemma}

\begin{lemma}
\label{lowerboundasymptotic}\
\[
\lambda_{1}(A+\alpha B)\geq\alpha\lambda_{1}(B)
\]

\end{lemma}

The interconnection network $B$ forms a bipartite graph.

\begin{lemma}
\label{bipartite} The largest eigenvalue of a bipartite graph $B=\left[
\begin{array}
[c]{cc}%
\mathbf{0} & B_{12}\\
B_{12}^{T} & \mathbf{0}%
\end{array}
\right]  $ follows $\lambda_{1}(B)=\sqrt{\lambda_{1}\left(  B_{12}^{T}
B_{12}\right)  }$ where $B_{12}$ is possibly asymmetric \cite{PVM_Spectrabook}.
\end{lemma}

\begin{lemma}
\label{Lemma_regular} When $G_{1}$ and $G_{2}$ are both regular graphs with
the same average degree $E[D]$ and when any two nodes from $G_{1}$ and $G_{2}$
respectively are randomly interconnected with probability $p_{I}$, the average
spectral radius of the interconnected networks follows
\[
E[\lambda_{1}(A+\alpha B)]=E[D]+\alpha Np_{I}
\]
if the interdependent connections are not sparse.
\end{lemma}

A dense Erd\H{o}s-R\'{e}nyi (ER) random network approaches a regular network
when $N$ is large. Lemma \ref{Lemma_regular}, thus, can be applied as well to
cases where both $G_{1}$ and $G_{2}$ are dense ER random networks.

For any two commuting matrices $A$ and $B$, thus $AB=BA$, $\lambda
_{1}(A+B)=\lambda_{1}(A)+\lambda_{1}(B)$ \cite{PVM_Spectrabook}. This property
of commuting matrices makes the following two special cases analytically tractable.

\begin{lemma}
\label{aaii} When $A+\alpha B=\left[
\begin{array}
[c]{cc}%
A_{1} & \mathbf{0}\\
\mathbf{0} & A_{1}%
\end{array}
\right]  +\alpha\left[
\begin{array}
[c]{cc}%
\mathbf{0} & I\\
I & \mathbf{0}%
\end{array}
\right]  $, i.e., the interconnected networks are composed of two identical
networks, where one network is indexed from $1$ to $N$ and the other from
$N+1$ to $2N$, with an interconnecting link between each so-called image node
pair $(i,N+i)$ from the two individual networks respectively, its largest
eigenvalue $\lambda_{1}(A+\alpha B)=\lambda_{1}(A)+\alpha$.
\end{lemma}

\begin{proof}
When $A+\alpha B=\left[
\begin{array}
[c]{cc}%
A_{1} & \mathbf{0}\\
\mathbf{0} & A_{1}%
\end{array}
\right]  +\alpha\left[
\begin{array}
[c]{cc}%
\mathbf{0} & I\\
I & \mathbf{0}%
\end{array}
\right]  ,$ matrix $A$ and $\alpha B$ are commuting
\[
A\cdot\alpha B=\alpha\left[
\begin{array}
[c]{cc}%
\mathbf{0} & A_{1}\\
A_{1} & \mathbf{0}%
\end{array}
\right]  =\alpha BA
\]
Therefore, $\lambda_{1}(A+\alpha B)=\lambda_{1}(A)+\lambda_{1}(\alpha
B)=\lambda_{1}(A_{1})+\alpha\lambda_{1}(B)$. The network $B$ is actually a set
of isolated links. Hence, $\lambda_{1}(B)=1$.
\end{proof}

\begin{lemma}
\label{aaaa}When $A+\alpha B=\left[
\begin{array}
[c]{cc}%
A_{1} & \mathbf{0}\\
\mathbf{0} & A_{1}%
\end{array}
\right]  +\alpha\left[
\begin{array}
[c]{cc}%
\mathbf{0} & A_{1}\\
A_{1} & \mathbf{0}%
\end{array}
\right]  $, its largest eigenvalue $\lambda_{1}(A+\alpha B)=\left(
1+\alpha\right)  \lambda_{1}(A_{1})$.
\end{lemma}

\begin{proof}
When $A+\alpha B=\left[
\begin{array}
[c]{cc}%
A_{1} & \mathbf{0}\\
\mathbf{0} & A_{1}%
\end{array}
\right]  +\alpha\left[
\begin{array}
[c]{cc}%
\mathbf{0} & A_{1}\\
A_{1} & \mathbf{0}%
\end{array}
\right]  ,$ matrix $A$ and $\alpha B$ are commuting
\[
A\cdot\alpha B=\alpha\left[
\begin{array}
[c]{cc}%
\mathbf{0} & A_{1}^{2}\\
A_{1}^{2} & \mathbf{0}%
\end{array}
\right]  =\alpha BA
\]
Therefore $\lambda_{1}(A+\alpha B)=\lambda_{1}(A)+\lambda_{1}(\alpha
B)=\left(  1+\alpha\right)  \lambda_{1}(A)=\left(  1+\alpha\right)
\lambda_{1}(A_{1})$.
\end{proof}

When $A$ and $B$ are not commuting, little can be known about the eigenvalues
of $\lambda_{1}(A+\alpha B)$, given the spectrum of $A$ and of $B$. For
example, even when the eigenvalue of $A$ and $B$ are known and bounded, the
largest eigenvalue of $\lambda_{1}(A+\alpha B)$ can be unbounded
\cite{PVM_Spectrabook}.

\subsection{Lower bounds for $\lambda_{1}\left(  A+\alpha B\right)  $}

We now denote matrix $A+\alpha B$ to be $W$. Applying the Rayleigh inequality
\cite[p. 223]{PVM_Spectrabook} to the symmetric matrix $W=A+\alpha B$ yields
\[
\frac{z^{T}Wz}{z^{T}z}\leq\lambda_{1}\left(  W\right)
\]
where equality holds only if $z$ is the principal eigenvector of $W$.

\begin{theorem}
\label{theorem_lowerbound1} The best possible lower bound $\frac{z^{T}Wz}
{z^{T}z}$ of interdependent networks $W$ by choosing $z$ as the linear
combination of $x$ and $y$, the largest eigenvector of $A_{1}$ and $A_{2}$
respectively, is
\begin{equation}
\lambda_{1}\left(  W\right)  \geq\max\left(  \lambda_{1}\left(  A_{1}\right)
,\lambda_{1}\left(  A_{2}\right)  \right)  +\left(  \sqrt{\left(
\frac{\lambda_{1}\left(  A_{1}\right)  -\lambda_{1}\left(  A_{2}\right)  }
{2}\right)  ^{2}+\xi^{2}}-\left\vert \frac{\lambda_{1}\left(  A_{1}\right)
-\lambda_{1}\left(  A_{2}\right)  }{2}\right\vert \right)  \label{lowerbound1}%
\end{equation}
where $\xi=\alpha x^{T}B_{12}y$.
\end{theorem}

When $\alpha=0$, the lower bound becomes the exact solution $\lambda
_{1}\left(  W\right)  =\lambda_{L}$. When the two individual networks have the
same largest eigenvalue $\lambda_{1}\left(  A_{1}\right)  =\lambda_{1}\left(
A_{2}\right)  $, we have
\[
\lambda_{1}\left(  W\right)  \geq\lambda_{1}\left(  A_{1}\right)  +\alpha
x^{T}B_{12}y
\]

\begin{theorem}
\label{theorem_lowerbound2} The best possible lower bound $\lambda_{1}%
^{2}\left(  W\right)  \geq\frac{z^{T}W^{2}z}{z^{T}z}$ by choosing $z$ as the
linear combination of $x$ and $y$, the largest eigenvector of $A_{1}$ and
$A_{2}$ respectively, is
\begin{align}
\lambda_{1}^{2}\left(  W\right)   &  \geq\frac{\left(  \lambda_{1}^{2}\left(
A_{1}\right)  +\alpha^{2}\left\Vert B_{12}^{T}x\right\Vert _{2}^{2}
+\lambda_{1}^{2}\left(  A_{2}\right)  +\alpha^{2}\left\Vert B_{12}y\right\Vert
_{2}^{2}\right)  }{2}+\label{lowerbound2}\\
&  \sqrt{\left(  \frac{\lambda_{1}^{2}\left(  A_{1}\right)  +\alpha
^{2}\left\Vert B_{12}^{T}x\right\Vert _{2}^{2}-\lambda_{1}^{2}\left(
A_{2}\right)  -\alpha^{2}\left\Vert B_{12}y\right\Vert _{2}^{2}}{2}\right)
^{2}+\theta^{2}}\nonumber
\end{align}
where $\theta=\alpha\left(  \lambda_{1}\left(  A_{1}\right)  +\lambda
_{1}\left(  A_{2}\right)  \right)  x^{T}B_{12}y$.
\end{theorem}

In general,
\[
\frac{z^{T}W^{k}z}{z^{T}z}\leq\lambda_{1}^{k}\left(  W\right)
\]
The largest eigenvalue is lower bounded by
\[
\left(  \frac{z^{T}W^{k}z}{z^{T}z}\right)  ^{1/k}\leq\lambda_{1}\left(
W\right)
\]

\begin{theorem}
\label{theorem_higherorder} Given a vector $z$, $\left(  \frac{z^{T}W^{s}
z}{z^{T}z}\right)  ^{1/s}\leq\left(  \frac{z^{T}W^{k}z}{z^{T}z}\right)
^{1/k}$when $k$ is an even integer and $0<s<k$. Furthermore,
\[
lim_{k\rightarrow\infty}\left(  \frac{z^{T}W^{k}z}{z^{T}z}\right)
^{1/k}=\lambda_{1}\left(  W\right)
\]

\end{theorem}

Hence, given a vector $z$, we could further improve the lower bound $\left(
\frac{z^{T}W^{k}z}{z^{T}z}\right)  ^{1/k}$ by taking a higher even power $k$.
Note that Theorem \ref{theorem_lowerbound1} and \ref{theorem_lowerbound2}
express the lower bound as a function of component network $A_{1}$, $A_{2}$
and $B$ and their eigenvalues/eigenvectors, which illustrates the effect of
component network features on the epidemic characterizer $\lambda_{1}\left(
W\right)  $.

\subsection{Upper bound for $\lambda_{1}\left(  A+\alpha B\right)  $}

\begin{theorem}
\label{theorem_upperbound} The largest eigenvalue of interdependent networks
$\lambda_{1}(W)$ is upper bounded by
\begin{align}
\lambda_{1}(W)  &  \leq\max{(\lambda_{1}(A_{1}),\lambda_{1}(A_{2}))}%
+\alpha\lambda_{1}(B)\label{upperbound}\\
&  =\max{(\lambda_{1}(A_{1}),\lambda_{1}(A_{2}))}+\alpha\sqrt{\lambda
_{1}(B_{12}B_{12}^{T})}%
\end{align}

\end{theorem}

This upper bound is reached when the principal eigenvector of $B_{12}%
B_{12}^{T}$ coincides with the principal eigenvector of ${A_{1}}$ if
${\lambda}_{1}{(A_{1})\geq\lambda}_{1}{(A_{2})}$ and when the principal
eigenvector of $B_{12}^{T}B_{12}$ coincides with the principal eigenvector of
${A_{2}}$ if ${\lambda}_{1}{(A_{1})\leq\lambda}_{1}{(A_{2}).}$

\subsection{Perturbation analysis for small and large $\alpha$}

In this subsection, we derive the perturbation approximation of $\lambda
_{1}(W)$ for small and large $\alpha$, respectively, as a function of
component networks and their eigenvalues/eigenvectors.

We start with small $\alpha$ cases. The problem is to find the largest
eigenvalue $\sup_{z\neq0}\frac{z^{T}Wz}{z^{T}z}$ of $W$, with the condition
that
\[
\left\{
\begin{array}
[c]{c}%
(W-\lambda I)z=0\\
z^{T}z=1
\end{array}
\right.
\]
When the solution is analytical in $\alpha$, we express $\lambda$ and $z$ by
Taylor expansion as
\begin{align*}
\lambda &  =\sum_{k=0}^{\infty}\lambda^{(k)}\alpha^{k}\\
z  &  =\sum_{k=0}^{\infty}z^{(k)}\alpha^{k}%
\end{align*}
Substituting the expansion in the eigenvalue equation gives
\[
(A+\alpha B)\sum_{k=0}^{\infty}z^{(k)}\alpha^{k}=\sum_{k=0}^{\infty}
\lambda^{(k)}\alpha^{k}\sum_{k=0}^{\infty}z^{(k)}\alpha^{k}
\]
where all the coefficients of $\alpha^{k}$ on the left must equal those on the
right. Performing the products and reordering the series we obtain
\[
\sum_{k=0}^{\infty}\left(  Az^{(k)}+Bz^{(k-1)}\right)  \alpha^{k}=\sum
_{k=0}^{\infty}\left(  \sum_{i=0}^{k}\lambda^{(k-i)}z^{(i)}\right)
\alpha^{k}
\]
This leads to a hierarchy of equations
\[
Az^{(k)}+Bz^{(k-1)}=\sum_{i=0}^{k}\lambda^{(k-i)}z^{(i)}
\]
The same expansion must meet the normalization condition
\[
z^{T}z=1
\]
or equivalently,
\[
\left(  \sum_{k=0}^{\infty}z^{(k)}\alpha^{k},\sum_{j=0}^{\infty}z^{(j)}
\alpha^{j}\right)  =1
\]
where $(u,v)=\sum_{i}u_{i}v_{i}$ represents the scalar product. The
normalization condition leads to a set of equations
\begin{equation}
\sum_{i=1}^{k}\left(  z^{(k-i)},z^{(i)}\right)  =0 \label{AAA}%
\end{equation}
for any $k\geq1$ and $\left(  z^{(0)},z^{(0)}\right)  =1$.

Let $\lambda_{1}(A_{1})\left(  \lambda_{1}(A_{2})\right)  $ and $x(y)$ denote
the largest eigenvalue and the corresponding eigenvector of $A_{1}(A_{2})$
respectively. We examine two possible cases: (a) the non-degenerate case when
$\lambda_{1}(A_{1})>\lambda_{1}(A_{2})$ and (b) the degenerate case when
$\lambda_{1}(A_{1})=\lambda_{1}(A_{2})$ and the case $\lambda_{1}
(A_{1})<\lambda_{1}(A_{2})$ is equivalent to the first.

\begin{theorem}
\label{theorem_perturbation_nondegenerate} For small $\alpha$, in the
non-degenerate case, thus when $\lambda_{1}(A_{1})>\lambda_{1}(A_{2})$,
\begin{equation}
\lambda_{1}(W)=\lambda_{1}(A_{1})+\alpha^{2}(x^{(0)})^{T}B_{12}\left(
\lambda_{1}(A_{1})I-A_{2}\right)  ^{-1}B_{12}^{T}x^{(0)}+O(\alpha^{3})
\label{perturbation_non}%
\end{equation}
where $\left(  z^{(0)}\right)  ^{T}=\left(
\begin{array}
[c]{cc}%
x^{T} & \mathbf{0}^{T}%
\end{array}
\right)  $.
\end{theorem}

Note that in (\ref{secondorder_A}) $B$ is symmetric and $\left(  \lambda
^{(0)}I-A\right)  $ is positive definite and so is $B\left(  \lambda
^{(0)}I-A\right)  ^{-1}B$. Hence, this second order correction $\lambda^{(2)}
$ is always positive.

\begin{theorem}
\label{theorem_perturbation_degenerate} For small $\alpha$, when the two
component networks have the same largest eigenvalue $\lambda_{1}
(A_{1})=\lambda_{1}(A_{2})$,
\begin{align}
&  \lambda_{1}(W)=\lambda_{1}(A_{1})+\frac{1}{2}\alpha x^{T}B_{12}y+\alpha
^{2}(y^{(0)})^{T}B_{12}^{T}(\lambda^{(0)}I-A_{1}+x^{(0)}(x^{(0)})^{T}
)^{-1}\cdot\label{degenerate}\\
&  (B_{12}y^{(0)}-\lambda^{(1)}x^{(0)}+(x^{(0)})^{T}B_{12}(\lambda^{(0)}
I-A_{2}+x^{(0)}(x^{(0)})^{T})^{-1}B_{12}^{T}x^{(0)}-\lambda^{(1)}
y^{(0)})+O(\alpha^{3})\nonumber
\end{align}

\end{theorem}

In the degenerate case, the first order correction is positive and the slope
depends on $B_{12}$, $y$, and $x$. When $A_{1}$ and $A_{2}$ are identical, the
largest eigenvalue of the interdependent networks becomes
\[
\lambda=\lambda_{1}(A_{1})+\alpha\left(  B_{12}x,x\right)  +O(\alpha^{2})
\]
When $A=\left[
\begin{array}
[c]{cc}%
A_{1} & \mathbf{0}\\
\mathbf{0} & A_{1}%
\end{array}
\right]  $ and $B=\left[
\begin{array}
[c]{cc}%
\mathbf{0} & I\\
I & \mathbf{0}%
\end{array}
\right]  $, our result (\ref{degenerate}) in the degenerate case up to the
first order leads to $\lambda_{1}(A+\alpha B)=\lambda_{1}(A)+\alpha$, which is
an alternate proof of Lemma \ref{aaii}. When $A=\left[
\begin{array}
[c]{cc}%
A_{1} & \mathbf{0}\\
\mathbf{0} & A_{1}%
\end{array}
\right]  $ and $B=\left[
\begin{array}
[c]{cc}%
\mathbf{0} & A_{1}\\
A_{1} & \mathbf{0}%
\end{array}
\right]  $, (\ref{degenerate}) again explains Lemma \ref{aaaa} that
$\lambda_{1}(A+\alpha B)=\left(  1+\alpha\right)  \lambda_{1}(A_{1})$.

\begin{lemma}
\label{alphainfinity} For large $\alpha$, the spectral radius of
interconnected networks is
\begin{equation}
\lambda_{1}(A+\alpha B)=\alpha\lambda_{1}(B)+v^{T}Av+O\left(  \alpha
^{-1}\right)  \label{perturb_alarge}%
\end{equation}
where $v$ is the eigenvector belonging to $\lambda_{1}(B)$ and
\[
\lambda_{1}(A+\alpha B)\leq\lambda_{1}(A)+\alpha\lambda_{1}(B)+O\left(
\alpha^{-1}\right)
\]

\end{lemma}

\begin{proof}
The lemma \ref{alphainfinity} follows by applying perturbation theory
\cite{Wilkinson} to the matrix $\alpha\left(  B+\frac{1}{\alpha}A\right)  $
and the Rayleigh principle \cite{PVM_Spectrabook}, which states that
$v^{T}Av\leq\lambda_{1}(A)$, for any normalized vector $v$ such that
$v^{T}v=1$, with equality only if $v$ is the eigenvector belonging to the
eigenvalue $\lambda_{1}(A)$.
\end{proof}

\section{Numerical simulations}

\label{simulation}

In this section, we employ numerical calculations to quantify to what extent
the perturbation approximation (\ref{perturbation_non}) and (\ref{degenerate})
for small $\alpha$, the perturbation approximation (\ref{perturb_alarge}) for
large $\alpha$, the upper (\ref{upperbound}) and lower bound
(\ref{lowerbound2}) are close to the exact value $\lambda_{1}(W)=\lambda
_{1}(A+\alpha B)$. We investigate the condition under which the approximations
provide better estimates. The analytical results derived earlier are valid for
arbitrary interconnected network structures. For simulations, we consider two
classic network models as possible topologies of $G_{1}$ and $G_{2}$: (i) the
Erd\H{o}s-R\'{e}nyi (ER) random network \cite{ER1,ER2,ER3} and (ii) the
Barab\'{a}si-Albert (BA) scale-free network \cite{BA}. ER networks are
characterized by a Binomial degree distribution $\Pr[D=k]=\binom{N-1}{k}%
p^{k}(1-p)^{N-1-k}$, where $N$ is the size of the network and $p$ is the
probability that each node pair is randomly connected. In scale-free networks,
the degree distribution is given by a power law $\Pr[D=k]=ck^{-\lambda}$ such
that $\sum_{k=1}^{N-1}ck^{-\lambda}=1$ and $\lambda=3$ in BA scale-free networks.

In numerical simulations, we consider $N_{1}=N_{2}=1000$. Specifically, in the
BA scale-free networks $m=3$ and the corresponding link density is
$p_{BA}\simeq0.006$. We consider ER networks with the same link density
$p_{ER}=p_{BA}=0.006$. A coupled network $G$ is the union of $G_{1}$ and
$G_{2}$, which are chosen from the above mentioned models, together with
random interconnection links with density $p_{I}$, the probability that any
two nodes from $G_{1}$ and $G_{2}$ respectively are interconnected. Given the
network models of $G_{1}$ and $G_{2}$ and the interacting link density $p_{I}%
$, we generate 100 interconnected network realizations. For each realization,
we compute the spectral radius $\lambda_{1}(W)$, its perturbation
approximation (\ref{perturbation_non}) and (\ref{degenerate}) for small
$\alpha$, the perturbation approximation (\ref{perturb_alarge}) for large
$\alpha$, upper bound (\ref{upperbound}) and lower bound (\ref{lowerbound2})
for any $\alpha$. We compare their averages over the 100 coupled network
realizations. We investigate the degenerate case $\lambda_{1}(G_{1}%
)=\lambda_{1}(G_{2})$ where the largest eigenvalue of $G_{1}$ and $G_{2}$ are
the same and the non-degenerate case where $\lambda_{1}(G_{1})\neq\lambda
_{1}(G_{2})$ respectively.

\subsection{Non-degenerate case}%

\begin{figure}[h]%
\centering
\includegraphics[
height=2.7397in,
width=6.922in
]%
{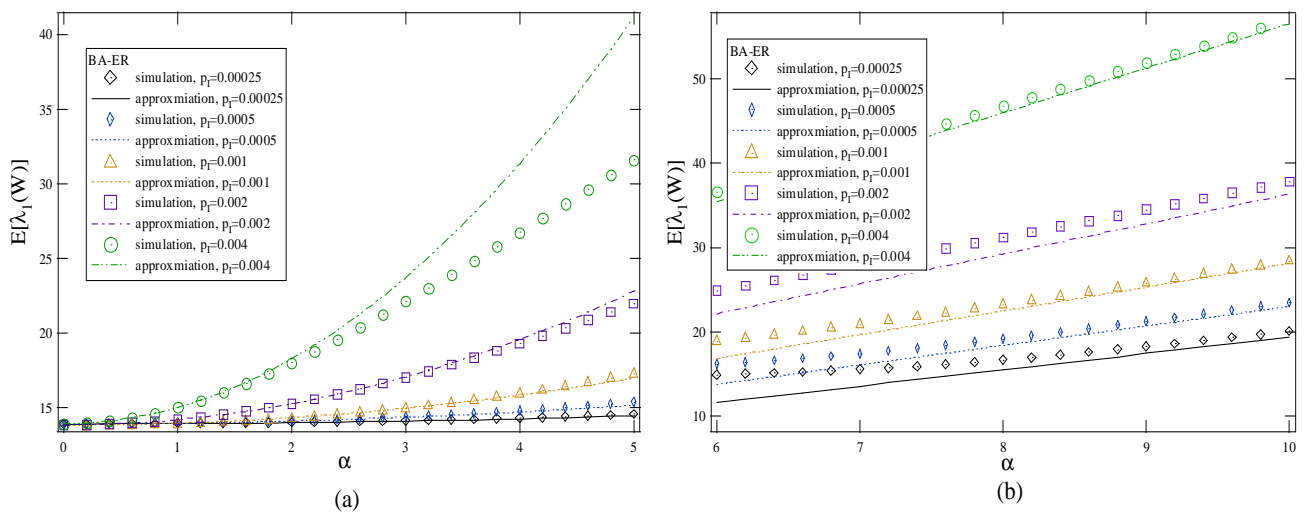}%
\caption{A plot of $\lambda_{1}(W)$ as a function of $\alpha$ for both
simulation results (symbol) and its (a) perturbation approximation
(\ref{perturbation_non}) for small $\alpha$ (dashed line) and (b) perturbation
approximation (\ref{perturb_alarge}) for large $\alpha$ (dashed line). The
interconneted network is composed of an ER random network and a BA scale-free
network both with $N=1000$ and link density $p=0.006$, randomly interconnected
with density $p_{I}$. All the results are averages of 100 realizations.}%
\label{baerapprox}%
\end{figure}
We consider the non-degenerate case in which $G_{1}$ is a BA scale-free
network with $N=1000,m=3$, $G_{2}$ is an ER random network with the same size
and link density $p_{ER}=p_{BA}\simeq0.006$, and the two networks are randomly
interconnected with link density $p_{I}$. We compute the largest average
eigenvalue $E[\lambda_{1}(W)]$ and the average of the perturbation
approximations and bounds mentioned above over 100 interconnected network
realizations for each interconnection link density $p_{I}\in\lbrack
0.00025,0.004]$ such that the average number of interdependent links ranges
from $\frac{N}{4},\frac{N}{2},N,2N$ to $4N$ and for each value $\alpha$ that
ranges from $0$ to $10$ with step size $0.05$.

For a single BA scale-free network, where the power exponent $\beta=3>2.5$,
the largest eigenvalue is $\left(  1+o(1)\right)  \sqrt{d_{\max}}$ where
$d_{\max}$ is the maximum degree in the network \cite{powerlaw}. The spectral
radius of a single ER random graph is close to the average degree
$(N-1)p_{ER}$ when the network is not sparse. When $p_{I}=0$, $\lambda
_{1}(G)=\max\left(  \lambda_{1}(G_{ER}),\lambda_{1}(G_{BA})\right)
=\lambda_{1}(G_{BA})>\lambda_{1}(G_{ER})$. The perturbation approximation is
expected to be close to the exact $\lambda_{1}(W)$ only for $\alpha
\rightarrow0$ and $\alpha\rightarrow\infty$. However, as shown in
Fig.~\ref{baerapprox}(a), the perturbation approximation for small $\alpha$
approximates $\lambda_{1}(W)$ well for a relative large range of $\alpha$,
especially for sparser interconnections, i.e., for a smaller interconnection
density $p_{I}$. Figure~\ref{baerapprox}(b) shows that the exact spectral
radius $\lambda_{1}(W)$ is already close to the large $\alpha$ perturbation
approximation, at least for $\alpha>8$.%

\begin{figure}[h]%
\centering
\includegraphics[
height=2.4716in,
width=6.6513in
]%
{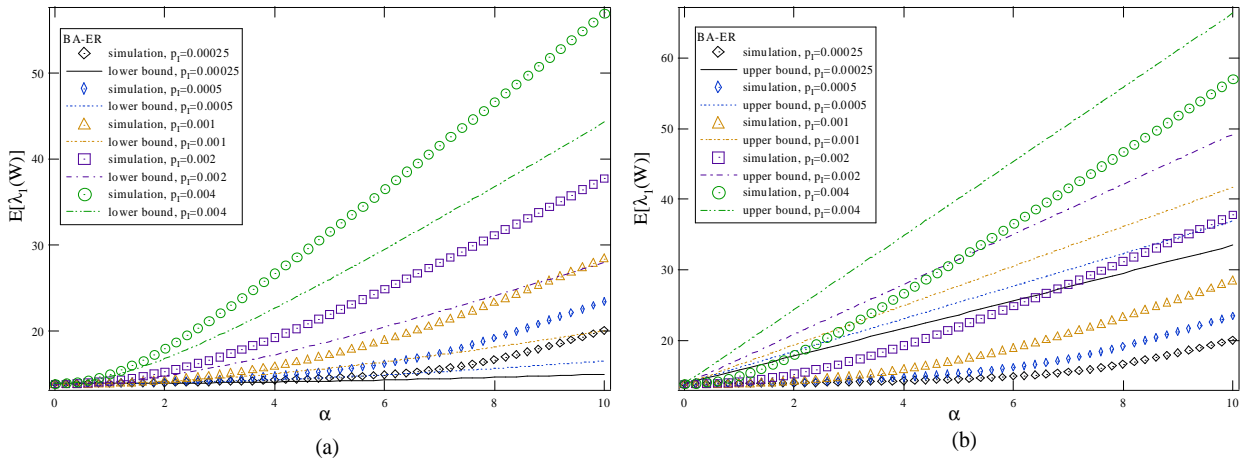}%
\caption{Plot $\lambda_{1}(W)$ as a function of $\alpha$ for both simulation
results (symbol) and its (a) its lower bound (\ref{lowerbound2}) (dashed line)
and (b) upper bound (\ref{upperbound})(dashed line). The interconneted network
is composed of an ER random network and a BA scale-free network both with
$N=1000$ and link density $p=0.006$, randomly interconnected with density
$p_{I}$. All the results are averages of 100 realizations.}%
\label{baerbounds}%
\end{figure}

As depicted in Fig.~\ref{baerbounds}, the lower bound (\ref{lowerbound2}) and
upper bound (\ref{upperbound}) are sharp, i.e., close to $\lambda_{1}(W)$ for
small $\alpha$. The lower and upper bounds are the same as $\lambda_{1}(W)$
when $\alpha\rightarrow0$. For large $\alpha$, the lower bound better
approximates $\lambda_{1}(W)$ when the interconnections are sparser. Another
lower bound $\alpha\lambda_{1}(B)\leq\lambda_{1}(W)$, i.e., Lemma
\ref{lowerboundasymptotic}, is sharp for large $\alpha$, as shown in
Fig.~\ref{lowerasymp}, especially for sparse interconnections. We do not
illustrate the lower bound (\ref{lowerbound1}) because the lower bound
(\ref{lowerbound2}) is always sharper or equally good. The lower bound
$\alpha\lambda_{1}(B)$ considers only the largest eigenvalue of the
interconnection network $B$ and ignores the two individual networks $G_{1}$
and $G_{2}$. The difference $\lambda_{1}(W)-$ $\alpha\lambda_{1}%
(B)=v^{T}Av+O\left(  \alpha^{-1}\right)  $ according to the large $\alpha$
perturbation approximation, is shown in Fig.~\ref{lowerasymp} to be larger for
denser interconnections. It suggests that $G_{1}$ and $G_{2}$ contribute more
to the spectral radius of the interconnected networks when the
interconnections are denser in this non-degenerate case. For large $\alpha$,
the upper bound is sharper when the interconnections are denser or when
$p_{I}$ is larger, as depicted in Figure \ref{baerbounds}(b). This is because
$\alpha\lambda_{1}(B)\leq\lambda_{1}(W)\leq\alpha\lambda_{1}(B)+\max
{(\lambda_{1}(A_{1}),\lambda_{1}(A_{2}))}$. When the interconnections are
sparse, $\lambda_{1}(W)$ is close to the lower bound $\alpha\lambda_{1}(B)$
and hence far from the upper bound.%
\begin{figure}[h]%
\centering
\includegraphics[
height=2.9274in,
width=4.4685in
]%
{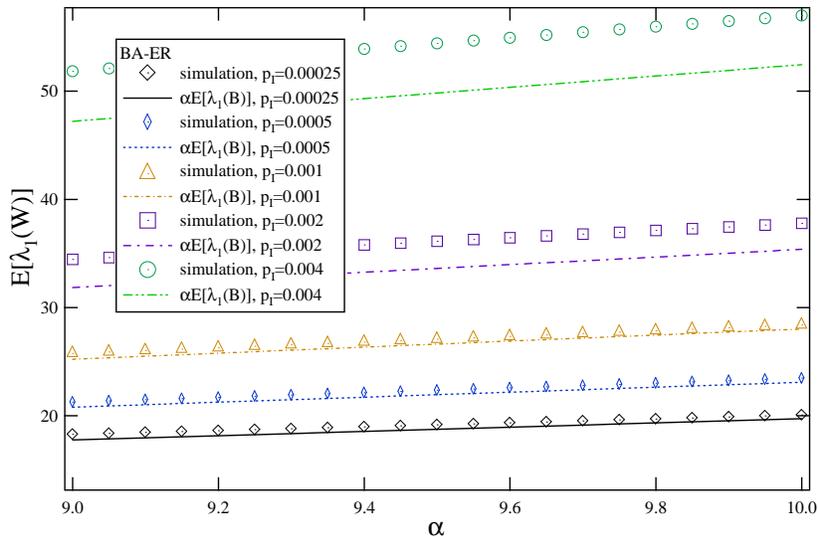}%
\caption{Plot $\lambda_{1}(W)$ as a function of $\alpha$ for both simulation
results (symbol) and its lower bound $\alpha\lambda_{1}(B)$ (dashed line). The
interconneted network is composed of an ER random network and a BA scale-free
network both with $N=1000$ and link density $p=0.006$, randomly interconnected
with density $p_{I}$. All the results are averages of 100 realizations.}%
\label{lowerasymp}%
\end{figure}

Most interdependent or coupled networks studied so far assume that both
individual networks have the same number of nodes $N$ and that the two
networks are interconnected randomly by $N$ one-to-one interconnections, or by
a fraction $q$ of the $N\ $one-to-one interconnections where $0<q\leq1$
\cite{Buldyrev2010,Huang_inter,Li}. These coupled networks correspond to our
sparse interconnection cases where $p_{I}\leq1$, when $\lambda_{1}(B)$ is well
approximated by the perturbation approximation for both small and large
$\alpha$. The spectral radius $\lambda_{1}(W)$ increases quadratically with
$\alpha$ for small $\alpha$, as described by the small $\alpha$ perturbation
approximation. The increase accelerates as $\alpha$ increases and converges to
a linear increase with $\alpha$, with slope $\lambda_{1}(B)$. Here we show the
cases in which $G_{1}$, $G_{2}$, and the interconnections are sparse, as in
most real-world networks. However, all the analytical results can be applied
to arbitrary interconnected network structures.

\subsection{Degenerate case}%

\begin{figure}[h]%
\centering
\includegraphics[
height=2.4664in,
width=6.7758in
]%
{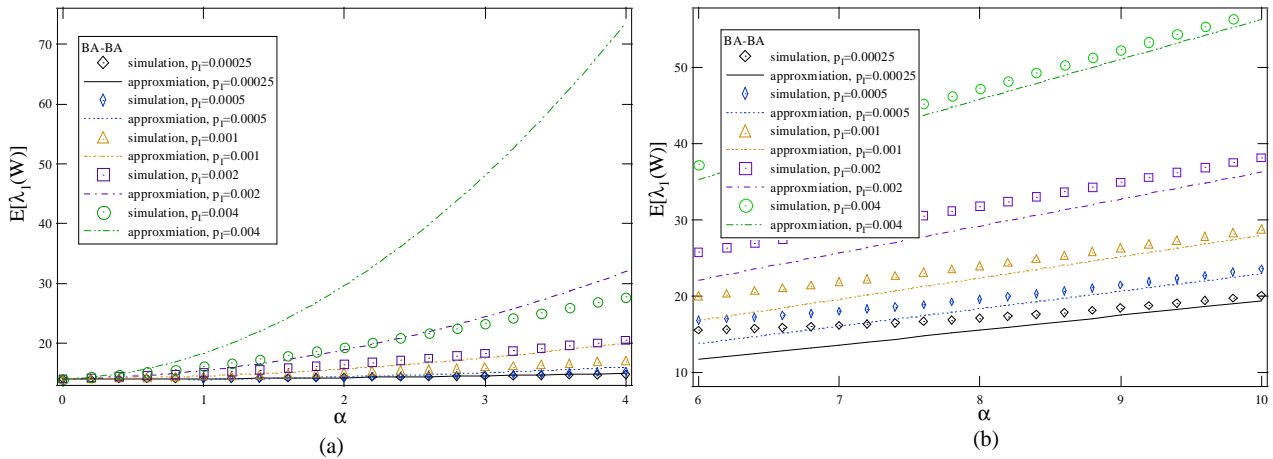}%
\caption{A plot of $\lambda_{1}(W)$ as a function of $\alpha$ for both
simulation results (symbol) and its (a) perturbation approximation
(\ref{degenerate}) for small $\alpha$ (dashed line) and (b) perturbation
approximation (\ref{perturb_alarge}) for large $\alpha$ (dashed line). The
interconneted network is composed of two identical BA scale-free networks with
$N=1000$ and link density $p=0.006$, randomly interconnected with density
$p_{I}$. All the results are averages of 100 realizations.}%
\label{babaasym}%
\end{figure}
We assume the spectrum \cite{Haemers} to be an unique fingerprint of a large
network. Two large networks of the same size seldom have the same largest
eigenvalue. Hence, most interconnected networks belong to the non-degenerate
case. Degenerate cases mostly occur when $G_{1}$ and $G_{2}$ are identical, or
when they are both regular networks with the same degree. We consider two
degenerate cases where both network $G_{1}$ and $G_{2}$ are ER random networks
or BA scale-free networks. Both ER and BA networks lead to the same
observations. Hence as an example we show the case in which both $G_{1}$ and
$G_{2}$ are BA scale-free networks of size $N=1000$ and both are randomly
interconnected with density $p_{I}\in\lbrack0.00025,0.004]$, as in the
non-degenerate case. Figure~\ref{babaasym}(a) shows that the perturbation
analysis well approximates $\lambda_{1}(W)$ for small $\alpha$, especially
when the interconnection density is small. Moreover, the small $\alpha$
perturbation approximation performs better in the non-degenerate case, i.e.,
is closer to $\lambda_{1}(W)$ than in degenerate cases [see
Fig.~\ref{baerapprox}(a)].%
\begin{figure}[h]%
\centering
\includegraphics[
height=2.5296in,
width=6.9427in
]%
{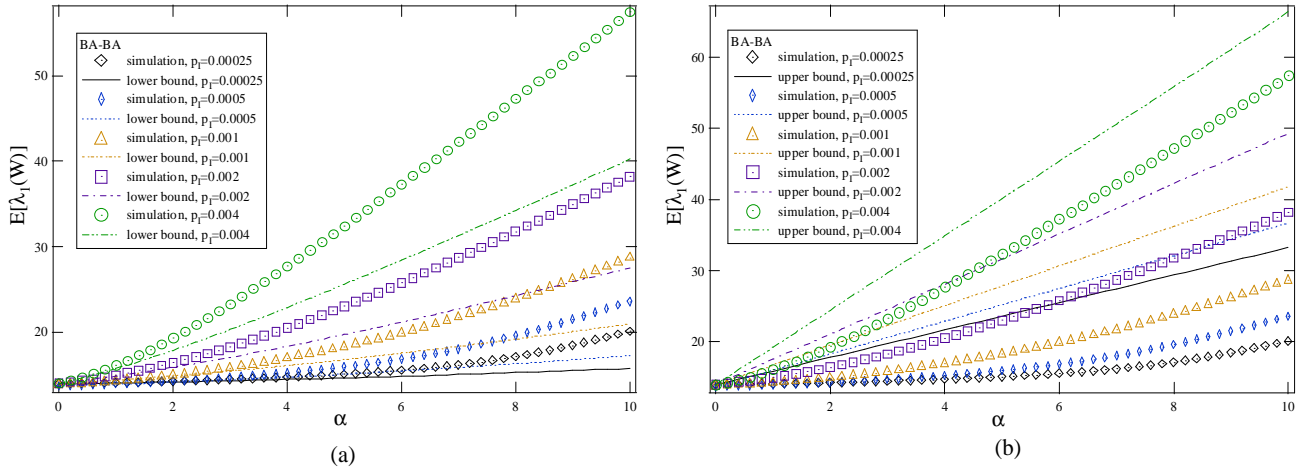}%
\caption{Plot $\lambda_{1}(W)$ as a function of $\alpha$ for both simulation
results (symbol) and its (a) its lower bound (\ref{lowerbound2}) (dashed line)
and (b) upper bound (\ref{upperbound})(dashed line). The interconneted network
is composed of two identical BA scale-free networks $N=1000$ and link density
$p=0.006$, randomly interconnected with density $p_{I}$. All the results are
averages of 100 realizations.}%
\label{bababounds}%
\end{figure}
Similarly, Fig.~\ref{bababounds} shows that both the lower and upper bound are
sharper for small $\alpha$. The lower bound better approximates $\lambda
_{1}(W)$ for sparser interconnections whereas the upper bound better
approximates $\lambda_{1}(W)$ for denser interconnections.

Thus far we have examined the cases where $G_{1}$, $G_{2}$, and the
interconnections are sparse, as is the case in most real-world networks.
However, if both $G_{1}$ and $G_{2}$ are dense ER random networks and if the
random interconnections are also dense, the upper bound is equal to
$\lambda_{1}(W)$, i.e., $\lambda_{1}(W)=\lambda_{1}(G_{1})+\alpha\lambda
_{1}(B)$ (see Lemma \ref{Lemma_regular}). Equivalently, the difference
$\lambda_{1}(W)-\alpha\lambda_{1}(B)$ is a constant $\lambda_{1}%
(G_{1})=\lambda_{1}(G_{2})$ independent of the interconnection density $p_{I}$.

In both the non-degenerate and degenerate case, $\lambda_{1}(W)$ is well
approximated by a perturbation analysis for a large range of small $\alpha$,
especially when the interconnections are sparse, and also for a large range of
large $\alpha$. The lower bound (\ref{lowerbound2}) and upper bound
(\ref{upperbound}) are sharper for small $\alpha$. Most real-world
interconnected networks are sparse and non-degenerate, where our perturbation
approximations are precise for a large range of $\alpha$, and thus reveal well
the effect of component network structures on the epidemic characterizer
$\lambda_{1}(W)$.

\section{Conclusion}

\label{conclusion}

We study interconnected networks that are composed of two individual networks
$G_{1}$ and $G_{2}$, and interconnecting links represented by adjacency
matricies $A_{1}$, $A_{2}$, and $B$ respectively. We consider SIS epidemic
spreading in these generic coupled networks, where the infection rate within
$G_{1}$ and $G_{2}$ is $\beta$, the infection rate between the two networks is
$\alpha\beta$, and the recovery rate is $\delta$ for all agents. Using a NIMFA
we show that the epidemic threshold with respect to $\beta/\delta$ is
$\tau_{c}=\frac{1}{\lambda_{1}(A+\alpha B)}$, where $A=\left[
\begin{array}
[c]{cc}%
A_{1} & \mathbf{0}\\
\mathbf{0} & A_{2}%
\end{array}
\right]  $ is the adjacency matrix of the two isolated networks $G_{1}$ and
$G_{2}$. The largest eigenvalue $\lambda_{1}(A+\alpha B)$ can thus be used to
characterize epidemic spreading. This eigenvalue $\lambda_{1}(A+\alpha B)$ of
a function of matrices seldom gives the contribution of each component
network. We analytically express the perturbation approximation for small and
large $\alpha$, lower and upper bounds for any $\alpha$, of $\lambda
_{1}(A+\alpha B)$ as a function of component networks $A_{1}$, $A_{2}$, and
$B$ and their largest eigenvalues/eigenvectors. Using numerical simulations,
we verify that these approximations or bounds approximate well the exact
$\lambda_{1}(A+\alpha B)$, especially when the interconnections are sparse and
when the largest eigenvalues of the two networks $G_{1}$ and $G_{2}$ are
different (the non-degenerate case), as is the case in most real-world
interconnected networks. Hence, these approximations and bounds reveal how
component network properties affect the epidemic characterizer $\lambda
_{1}(A+\alpha B)$. Note that the term $x^{T}B_{12}y$ contributes positively to
the perturbation approximation (\ref{degenerate}) and the lower bound
(\ref{lowerbound2}) of $\lambda_{1}(A+\alpha B)$ where $x$ and $y$ are the
principal eigenvector of network $G_{1}$ and $G_{2}$. This suggests that,
given two isolated networks $G_{1}$ and $G_{2}$, the interconnected networks
have a larger $\lambda_{1}(A+\alpha B)$ or a smaller epidemic threshold if the
two nodes $i$ and $j$ with a larger eigenvector component product $x_{i}y_{j}$
from the two networks, respectively, are interconnected. This observation
provides essential insights useful when designing interconnected networks to
be robust against epidemics. The largest eigenvalue also characterizes the
phase transition of coupled oscillators and percolation. Our results apply to
arbitrary interconnected network structures and are expected to apply to a
wider range of dynamic processes.

\section{Acknowledgements}

We wish to thank ONR (Grant N00014-09-1-0380, Grant N00014-12-1-0548), DTRA
(Grant HDTRA-1-10-1-0014, Grant HDTRA-1-09-1-0035), NSF (Grant CMMI 1125290),
the European EPIWORK, MULTIPLEX, CONGAS (Grant FP7-ICT-2011-8-317672), MOTIA
(Grant JLS-2009-CIPS-AG-C1-016) and LINC projects, the Deutsche
Forschungsgemeinschaft (DFG), the Next Generation Infrastructure (Bsik) and
the Israel Science Foundation for financial support.

\appendix

\section{Proofs}

\subsection{Proof of Lemma \ref{Lemma_regular}}

In any regular graph, the minimal and maximal node strength are both equal to
the average node strength. Since the largest eigenvalue is lower bounded by
the average node strength and upper bounded by the maximal node strength as
proved below in Lemma \ref{boundlamda}, a regular graph has the minimal
possible spectral radius, which equals the average node strength. When the
interdependent links are randomly connected with link density $p_{I}$, the
coupled network is asymptotically a regular graph with average node strength
$E[D]+\alpha Np_{I}$, if $p_{I}$ is a constant.

\begin{lemma}
\label{boundlamda}For any $N\times N$ weighted symmetric matrix $W$,
\[
E[S]\leq\lambda_{1}(W)\leq\max s_{r}
\]
where $s_{r}=\sum_{j=1}^{N}w_{rj}$ is defined as the node strength of node $r
$ and $E[S]$ is the average node strength over all the nodes in graph $G$.
\end{lemma}

\begin{proof}
The largest eigenvalue $\lambda_{1}$ follows
\[
\lambda_{1}=\sup_{x\neq0}\frac{x^{T}Wx}{x^{T}x}
\]
when matrix $W$ is symmetric and the maximum is attained if and only if $x$ is
the eigenvector of $W$ belonging to $\lambda_{1}(W)$. For any other vector
$y\neq x,$ it holds that $\lambda_{1}\geq\frac{y^{T}Wy}{y^{T}y}$. By choosing
the vector $y=u=(1,1,...,1),$ we have
\[
\lambda_{1}\geq\frac{1}{N}\sum_{i=1}^{N}\sum_{j=1}^{N}w_{ij}=\frac{1}{N}
\sum_{i=1}^{N}s_{i}=E[S]
\]
where $w_{ij}$ is the element in matrix $W$ and $E[S]$ is the average node
strength of the graph $G$. The upper bound is proved by Gerschgorin circle
theorem. Suppose component $r$ of eigenvector $x$ has the largest modulus. The
eigenvector can be always normalized such that
\[
x^{^{\prime}}=\left(  \frac{x_{1}}{x_{r}},\frac{x_{2}}{x_{r}},...,\frac
{x_{r-1}}{x_{r}},1,\frac{x_{r+1}}{x_{r}},...,\frac{x_{N}}{x_{r}}\right)
\]
where $\left\vert \frac{x_{j}}{x_{r}}\right\vert \leq1$ for all $j$. Equating
component $r$ on both sides of the eigenvalue equation $Wx^{^{\prime}
}=\lambda_{1}x^{^{\prime}}$ gives
\[
\lambda_{1}(W)=\sum_{j=1}^{N}w_{rj}\frac{x_{j}}{x_{r}}\leq\sum_{j=1}
^{N}\left\vert w_{rj}\frac{x_{j}}{x_{r}}\right\vert \leq\sum_{j=1}
^{N}\left\vert w_{rj}\right\vert =s_{r}
\]
when none of the elements of matrix $W$ are negative. Since any component of
$x$ may have the largest modulus, $\lambda_{1} (W)\leq\max s_{r}$.
\end{proof}

\subsection{Proof of Theorem \ref{theorem_lowerbound1}}

We consider the $2N\times1$ vector $z$ as $z^{T}=\left[
\begin{array}
[c]{cc}%
C_{1}x^{T} & C_{2}y^{T}%
\end{array}
\right]  $ the linear combination of the principal eigenvector $x$ and $y$ of
the two individual networks respectively, where $x^{T}x=1$, $y^{T}y=1$,
$C_{1}^{2}+C_{2}^{2}=1$ such that $z^{T}z=1$ and compute
\begin{align*}
z^{T}Wz  &  =\left[
\begin{array}
[c]{cc}%
C_{1}x^{T} & C_{2}y^{T}%
\end{array}
\right]  \left[
\begin{array}
[c]{cc}%
A_{1} & \alpha B_{12}\\
\alpha B_{12}^{T} & A_{2}%
\end{array}
\right]  \left[
\begin{array}
[c]{c}%
C_{1}x\\
C_{2}y
\end{array}
\right] \\
&  =C_{1}^{2}x^{T}A_{1}x+C_{2}^{2}y^{T}A_{2}y+2\alpha2C_{1}C_{2}x^{T}B_{12}y\\
&  =C_{1}^{2}\lambda_{1}\left(  A_{1}\right)  +C_{2}^{2}\lambda_{1}\left(
A_{2}\right)  +2C_{1}C_{2}\xi
\end{align*}
where $\xi=\alpha x^{T}B_{12}y$. By Rayleigh's principle $\lambda_{1}\left(
W\right)  \geq\frac{z^{T}Wz}{z^{T}z}=z^{T}Wz$. We could improve this lower
bound by selecting $z$ as the best linear combination ($C_{1}$ and $C_{2}$) of
$x$ and $y$. Let $\lambda_{L}$ be the best possible lower bound $\frac
{z^{T}Wz}{z^{T}z}$ via the optimal linear combination of $x$ and $y$. Thus,
\[
\lambda_{L}=\max_{C_{1}^{2}+C_{2}^{2}=1}C_{1}^{2}\lambda_{1}\left(
A_{1}\right)  +C_{2}^{2}\lambda_{1}\left(  A_{2}\right)  +2C_{1}C_{2}\xi
\]
We use the Lagrange multipliers method and define the Lagrange function as
\[
\Lambda=C_{1}^{2}\lambda_{1}\left(  A_{1}\right)  +C_{2}^{2}\lambda_{1}\left(
A_{2}\right)  +2C_{1}C_{2}\xi-\mu\left(  C_{1}^{2}+C_{2}^{2}-1\right)
\]
where $\mu$ is the Lagrange multiplier. The maximum is achieved at the
solutions of
\begin{align*}
\frac{\partial\Lambda}{\partial C_{1}}  &  =2C_{1}\lambda_{1}\left(
A_{1}\right)  +2C_{2}\xi-2C_{1}\mu=0\\
\frac{\partial\Lambda}{\partial C_{2}}  &  =2C_{2}\lambda_{1}\left(
A_{2}\right)  +2C_{1}\xi-2C_{2}\mu=0\\
\frac{\partial\Lambda}{\partial\mu}  &  =C_{1}^{2}+C_{2}^{2}-1=0
\end{align*}
Note that $\left(  C_{1}\frac{\partial\Lambda}{\partial C_{1}}+C_{2}
\frac{\partial\Lambda}{\partial C_{2}}\right)  /2=\lambda_{L}-\mu=0,$ which
leads to $\mu=\lambda_{L}.$ Hence, the maximum $\lambda_{L}$ is achieved at
the solution of
\begin{align*}
C_{1}\lambda_{1}\left(  A_{1}\right)  +C_{2}\xi-C_{1}\lambda_{L}  &  =0\\
C_{2}\lambda_{1}\left(  A_{2}\right)  +C_{1}\xi-C_{2}\lambda_{L}  &  =0
\end{align*}
that is
\[
\det\left(
\begin{array}
[c]{cc}%
\lambda_{1}\left(  A_{1}\right)  -\lambda_{L} & \xi\\
\xi & \lambda_{1}\left(  A_{2}\right)  -\lambda_{L}%
\end{array}
\right)  =0
\]
This leads to
\begin{align*}
\lambda_{L}  &  =\frac{\lambda_{1}\left(  A_{1}\right)  +\lambda_{1}\left(
A_{2}\right)  }{2}+\sqrt{\left(  \frac{\lambda_{1}\left(  A_{1}\right)
-\lambda_{1}\left(  A_{2}\right)  }{2}\right)  ^{2}+\xi^{2}}\\
&  =\frac{\lambda_{1}\left(  A_{1}\right)  +\lambda_{1}\left(  A_{2}\right)
}{2}+\left\vert \frac{\lambda_{1}\left(  A_{1}\right)  -\lambda_{1}\left(
A_{2}\right)  }{2}\right\vert +\left(  \sqrt{\left(  \frac{\lambda_{1}\left(
A_{1}\right)  -\lambda_{1}\left(  A_{2}\right)  }{2}\right)  ^{2}+\xi^{2}
}-\left\vert \frac{\lambda_{1}\left(  A_{1}\right)  -\lambda_{1}\left(
A_{2}\right)  }{2}\right\vert \right) \\
&  =\max\left(  \lambda_{1}\left(  A_{1}\right)  ,\lambda_{1}\left(
A_{2}\right)  \right)  +\left(  \sqrt{\left(  \frac{\lambda_{1}\left(
A_{1}\right)  -\lambda_{1}\left(  A_{2}\right)  }{2}\right)  ^{2}+\xi^{2}
}-\left\vert \frac{\lambda_{1}\left(  A_{1}\right)  -\lambda_{1}\left(
A_{2}\right)  }{2}\right\vert \right)
\end{align*}
The maximum is obtained when
\[
z^{T}=\pm\left[
\begin{array}
[c]{cc}%
\sqrt{\frac{\lambda_{1}\left(  A_{2}\right)  -\lambda_{L}}{\lambda_{1}\left(
A_{1}\right)  +\lambda_{1}\left(  A_{2}\right)  -2\lambda_{L}}}x^{T} &
\sqrt{\frac{\lambda_{1}\left(  A_{1}\right)  -\lambda_{L}}{\lambda_{1}\left(
A_{1}\right)  +\lambda_{1}\left(  A_{2}\right)  -2\lambda_{L}}}y^{T}%
\end{array}
\right]
\]

\subsection{Proof of Theorem \ref{theorem_lowerbound2}}

By Rayleigh's principle $\lambda_{1}^{2}\left(  W\right)  \geq\frac{z^{T}
W^{2}z}{z^{T}z}=z^{T}W^{2}z$. We consider $z$ as linear combination
$z^{T}=\left[
\begin{array}
[c]{cc}%
C_{1}x^{T} & C_{2}y^{T}%
\end{array}
\right]  $ of $x$ and $y$. The lower bound
\begin{align*}
z^{T}W^{2}z  &  =\left[
\begin{array}
[c]{cc}%
C_{1}x^{T} & C_{2}y^{T}%
\end{array}
\right]  \left[
\begin{array}
[c]{cc}%
A_{1}^{2}+\alpha^{2}B_{12}B_{12}^{T} & \alpha\left(  A_{1}B_{12}+B_{12}
A_{2}\right) \\
\alpha\left(  A_{1}B_{12}+B_{12}A_{2}\right)  ^{T} & A_{2}^{2}+\alpha
^{2}B_{12}^{T}B_{12}%
\end{array}
\right]  \left[
\begin{array}
[c]{c}%
C_{1}x\\
C_{2}y
\end{array}
\right] \\
&  =C_{1}^{2}x^{T}A_{1}^{2}x+C_{2}^{2}y^{T}A_{2}^{2}y+\alpha^{2}\left(
C_{1}^{2}x^{T}B_{12}B_{12}^{T}x+C_{2}^{2}y^{T}B_{12}^{T}B_{12}y\right)
+2\alpha C_{1}C_{2}x^{T}\left(  A_{1}B_{12}+B_{12}A_{2}\right)  y\\
&  =C_{1}^{2}\lambda_{1}^{2}\left(  A_{1}\right)  +C_{2}^{2}\lambda_{1}
^{2}\left(  A_{2}\right)  +2C_{1}C_{2}\theta+\alpha^{2}\left(  C_{1}
^{2}\left\Vert B_{12}^{T}x\right\Vert _{2}^{2}+C_{2}^{2}\left\Vert
B_{12}y\right\Vert _{2}^{2}\right)
\end{align*}
where $\theta=\alpha\left(  \lambda_{1}\left(  A_{1}\right)  +\lambda
_{1}\left(  A_{2}\right)  \right)  x^{T}B_{12}y$. Let $\lambda_{L}$ be the
best possible lower bound $z^{T}W^{2}z$ via the optimal linear combination
($C_{1}$ and $C_{2}$) of $x$ and $y$. Thus,
\[
\lambda_{L}=\max_{C_{1}^{2}+C_{2}^{2}=1}C_{1}^{2}\lambda_{1}^{2}\left(
A_{1}\right)  +C_{2}^{2}\lambda_{1}^{2}\left(  A_{2}\right)  +2C_{1}
C_{2}\theta+\alpha^{2}\left(  C_{1}^{2}\left\Vert B_{12}^{T}x\right\Vert
_{2}^{2}+C_{2}^{2}\left\Vert B_{12}y\right\Vert _{2}^{2}\right)
\]
We use the Lagrange multipliers method and define the Lagrange function as
\[
\Lambda=C_{1}^{2}\lambda_{1}^{2}\left(  A_{1}\right)  +C_{2}^{2}\lambda
_{1}^{2}\left(  A_{2}\right)  +2C_{1}C_{2}\theta+\alpha^{2}\left(  C_{1}
^{2}\left\Vert B_{12}^{T}x\right\Vert _{2}^{2}+C_{2}^{2}\left\Vert
B_{12}y\right\Vert _{2}^{2}\right)  -\mu\left(  C_{1}^{2}+C_{2}^{2}-1\right)
\]
where $\mu$ is the Lagrange multiplier. The maximum is achieved at the
solutions of
\begin{align*}
\frac{\partial\Lambda}{\partial C_{1}}  &  =2C_{1}\lambda_{1}^{2}\left(
A_{1}\right)  +2\alpha C_{2}\left(  \lambda_{1}\left(  A_{1}\right)
+\lambda_{1}\left(  A_{2}\right)  \right)  x^{T}B_{12}y+2\alpha^{2}
C_{1}\left\Vert B_{12}^{T}x\right\Vert _{2}^{2}-2C_{1}\mu=0\\
\frac{\partial\Lambda}{\partial C_{2}}  &  =2C_{2}\lambda_{1}^{2}\left(
A_{2}\right)  +2\alpha C_{1}\left(  \lambda_{1}\left(  A_{1}\right)
+\lambda_{1}\left(  A_{2}\right)  \right)  x^{T}B_{12}y+2\alpha^{2}
C_{2}\left\Vert B_{12}y\right\Vert _{2}^{2}-2C_{2}\mu=0\\
\frac{\partial\Lambda}{\partial\mu}  &  =C_{1}^{2}+C_{2}^{2}-1=0
\end{align*}
which lead to $\left(  C_{1}\frac{\partial\Lambda}{\partial C_{1}}+C_{2}
\frac{\partial\Lambda}{\partial C_{2}}\right)  /2=\lambda_{L}-\mu=0$. Hence,
the maximum $\lambda_{L}$ is achieved at the solution of
\begin{align*}
C_{1}\lambda_{1}^{2}\left(  A_{1}\right)  +C_{2}\theta+\alpha^{2}
C_{1}\left\Vert B_{12}^{T}x\right\Vert _{2}^{2}-C_{1}\lambda_{L}  &  =0\\
C_{2}\lambda_{1}^{2}\left(  A_{2}\right)  +C_{1}\theta+\alpha^{2}
C_{2}\left\Vert B_{12}y\right\Vert _{2}^{2}-C_{2}\lambda_{L}  &  =0
\end{align*}
that is
\[
\det\left(
\begin{array}
[c]{cc}%
\lambda_{1}^{2}\left(  A_{1}\right)  +\alpha^{2}\left\Vert B_{12}
^{T}x\right\Vert _{2}^{2}-\lambda_{L} & \theta\\
\theta & \lambda_{1}^{2}\left(  A_{2}\right)  +\alpha^{2}\left\Vert
B_{12}y\right\Vert _{2}^{2}-\lambda_{L}%
\end{array}
\right)  =0
\]
This leads to
\[
\lambda_{L}^{2}-\left(  \lambda_{1}^{2}\left(  A_{1}\right)  +\alpha
^{2}\left\Vert B_{12}^{T}x\right\Vert _{2}^{2}+\lambda_{1}^{2}\left(
A_{2}\right)  +\alpha^{2}\left\Vert B_{12}y\right\Vert _{2}^{2}\right)
\lambda_{L}+\left(  \lambda_{1}^{2}\left(  A_{1}\right)  +\alpha^{2}\left\Vert
B_{12}^{T}x\right\Vert _{2}^{2}\right)  \left(  \lambda_{1}^{2}\left(
A_{2}\right)  +\alpha^{2}\left\Vert B_{12}y\right\Vert _{2}^{2}\right)
-\theta^{2}=0
\]
Hence,
\begin{align*}
\lambda_{L}  &  =\frac{\left(  \lambda_{1}^{2}\left(  A_{1}\right)
+\alpha^{2}\left\Vert B_{12}^{T}x\right\Vert _{2}^{2}+\lambda_{1}^{2}\left(
A_{2}\right)  +\alpha^{2}\left\Vert B_{12}y\right\Vert _{2}^{2}\right)  }
{2}\\
&  +\frac{\sqrt{\left(  \lambda_{1}^{2}\left(  A_{1}\right)  +\alpha
^{2}\left\Vert B_{12}^{T}x\right\Vert _{2}^{2}+\lambda_{1}^{2}\left(
A_{2}\right)  +\alpha^{2}\left\Vert B_{12}y\right\Vert _{2}^{2}\right)
^{2}-4\left(  \left(  \lambda_{1}^{2}\left(  A_{1}\right)  +\alpha
^{2}\left\Vert B_{12}^{T}x\right\Vert _{2}^{2}\right)  \left(  \lambda_{1}
^{2}\left(  A_{2}\right)  +\alpha^{2}\left\Vert B_{12}y\right\Vert _{2}
^{2}\right)  -\theta^{2}\right)  }}{2}\\
&  =\frac{\left(  \lambda_{1}^{2}\left(  A_{1}\right)  +\alpha^{2}\left\Vert
B_{12}^{T}x\right\Vert _{2}^{2}+\lambda_{1}^{2}\left(  A_{2}\right)
+\alpha^{2}\left\Vert B_{12}y\right\Vert _{2}^{2}\right)  }{2}+\sqrt{\left(
\frac{\lambda_{1}^{2}\left(  A_{1}\right)  +\alpha^{2}\left\Vert B_{12}
^{T}x\right\Vert _{2}^{2}-\lambda_{1}^{2}\left(  A_{2}\right)  -\alpha
^{2}\left\Vert B_{12}y\right\Vert _{2}^{2}}{2}\right)  ^{2}+\theta^{2}}%
\end{align*}
which is obtained when
\begin{align*}
C_{1}  &  =\frac{\theta}{\sqrt{\theta^{2}+\left(  \lambda_{L}-\lambda_{1}
^{2}\left(  A_{1}\right)  -\alpha^{2}\left\Vert B_{12}^{T}x\right\Vert
_{2}^{2}\right)  ^{2}}}\\
C_{2}  &  =\frac{\lambda_{L}-\lambda_{1}^{2}\left(  A_{1}\right)  -\alpha
^{2}\left\Vert B_{12}^{T}x\right\Vert _{2}^{2}}{\sqrt{\theta^{2}+\left(
\lambda_{L}-\lambda_{1}^{2}\left(  A_{1}\right)  -\alpha^{2}\left\Vert
B_{12}^{T}x\right\Vert _{2}^{2}\right)  ^{2}}}%
\end{align*}

\subsection{Proof of Theorem \ref{theorem_higherorder}}

Any vector $z$ of size $2N$ with $zz^{T}=m$ can expressed as a linear
combination of the eigenvectors $(z_{1},z_{2},...,z_{2N})$ of matrix $W$
\[
\frac{z}{\sqrt{m}}=\sum_{i=1}^{2N}c_{i}z_{i}
\]
where $\sum_{i=1}^{2N}c_{i}^{2}=1$. Hence,
\begin{align*}
\frac{z^{T}W^{s}z}{z^{T}z}  &  =\left(  \sum_{i=1}^{2N}c_{i}z_{i}\right)
^{T}\left(  \sum_{i=1}^{2N}c_{i}W^{s}z_{i}\right) \\
&  =\left(  \sum_{i=1}^{2N}c_{i}z_{i}\right)  ^{T}\left(  \sum_{i=1}^{2N}
c_{i}\lambda_{i}^{s}z_{i}\right) \\
&  =\sum_{i=1}^{2N}c_{i}^{2}\lambda_{i}^{k}=\lambda_{1}^{s}\left(  \sum
_{i=1}^{2N}c_{i}^{2}\frac{\lambda_{i}^{s}}{\lambda_{1}^{k}}\right)
\end{align*}
Hence,
\[
lim_{k\rightarrow\infty}\left(  \frac{z^{T}W^{k}z}{z^{T}z}\right)
^{1/k}=\lambda_{1}\left(  W\right)
\]
According to Lyapunov's inequality,
\[
\left(  E\left[  \left\vert X\right\vert ^{s}\right]  \right)  ^{1/s}
\leq\left(  E\left[  \left\vert X\right\vert ^{t}\right]  \right)  ^{1/t}
\]
when $0<s<t$. Taking $\Pr[X=\frac{\lambda_{i}}{\lambda_{1}}]=c_{i}^{2}$, we
have
\[
\sum_{i=1}^{2N}c_{i}^{2}\frac{\lambda_{i}^{s}}{\lambda_{1}^{s}}\leq\sum
_{i=1}^{2N}c_{i}^{2}\left\vert \frac{\lambda_{i}}{\lambda_{1}}\right\vert
^{s}=\left(  E\left[  \left\vert X\right\vert ^{s}\right]  \right)  ^{1/s}
\leq\left(  E\left[  \left\vert X\right\vert ^{k}\right]  \right)  ^{1/k}
=\sum_{i=1}^{2N}c_{i}^{2}\frac{\lambda_{i}^{k}}{\lambda_{1}^{k}}
\]
since $k$ is even and $k>s>0$.

\subsection{Proof of Theorem \ref{theorem_upperbound}}%

\begin{align*}
\lambda_{1}\left(  W\right)   &  =\max_{x^{T}x+y^{T}y=1}\left[
\begin{array}
[c]{cc}%
x^{T} & y^{T}%
\end{array}
\right]  \left(  A+\alpha B\right)  \left[
\begin{array}
[c]{c}%
x\\
y
\end{array}
\right] \\
&  =\max_{x^{T}x+y^{T}y=1}\left(  \left[
\begin{array}
[c]{cc}%
x^{T} & y^{T}%
\end{array}
\right]  A\left[
\begin{array}
[c]{c}%
x\\
y
\end{array}
\right]  +\alpha\left[
\begin{array}
[c]{cc}%
x^{T} & y^{T}%
\end{array}
\right]  B\left[
\begin{array}
[c]{c}%
x\\
y
\end{array}
\right]  \right) \\
&  \leq\max_{x^{T}x+y^{T}y=1}\left(  x^{T}A_{1}x+y^{T}A_{2}y\right)
+\alpha\max_{x^{T}x+y^{T}y=1}\left[
\begin{array}
[c]{cc}%
x^{T} & y^{T}%
\end{array}
\right]  B\left[
\begin{array}
[c]{c}%
x\\
y
\end{array}
\right] \\
&  =\max\left(  {\lambda}_{1}{(A_{1}),\lambda}_{1}{(A_{2})}\right)
+\alpha\lambda_{1}(B)
\end{align*}
The inequality is due to the fact that the two terms are maximized
independently. The second term%
\begin{align*}
\lambda_{1}(B)  &  =\max_{x^{T}x+y^{T}y=1}\left(  x^{T}B_{12}y+y^{T}B_{12}
^{T}x\right) \\
&  =2\max_{x^{T}x+y^{T}y=1}x^{T}B_{12}y
\end{align*}
is equivalent to the system of equations%
\[
\left\{
\begin{array}
[c]{c}%
B_{12}y=\lambda_{1}(B)x\\
B_{12}y=\lambda_{1}(B)x\\
x^{T}x+y^{T}y=1
\end{array}
\right.
\]
or
\[
\left\{
\begin{array}
[c]{c}%
B_{12}^{T}B_{12}y=\lambda_{1}(B)^{2}y\\
B_{12}B_{12}^{T}x=\lambda_{1}(B)^{2}x\\
x^{T}x+y^{T}y=1
\end{array}
\right.
\]
which is to find the maximum eigenvalue (or more precisely the positive square
root) of the symmetric positive matrix $B_{12}B_{12}^{T}$
\[
\lambda_{1}(B)=\sqrt{\max_{x^{2}=1}x^{T}B_{12}B_{12}^{T}x}
\]
This actually proves Lemma \ref{bipartite}, the property $\lambda_{1}
(B)=\sqrt{\lambda_{1}(B_{12}B_{12}^{T})}$ of a bipartite graph $B.$

\subsection{Proof of Theorem \ref{theorem_perturbation_nondegenerate}}

The explicit expression up to the second order reads
\begin{equation}
\left(  A+\alpha B\right)  \left(  z^{(0)}+\alpha z^{(1)}+\alpha^{2}
z^{(2)}+O(\alpha^{3})\right)  =\left(  \lambda^{(0)}+\alpha\lambda
^{(1)}+\alpha^{2}\lambda^{(2)}+O(\alpha^{3})\right)  \left(  z^{(0)}+\alpha
z^{(1)}+\alpha^{2}z^{(2)}+O(\alpha^{3})\right)  \label{eigenequa}%
\end{equation}
The zero order expansion is simply
\[
Az^{(0)}=\lambda^{(0)}z^{(0)}
\]
The problem at zero order becomes to find the maximum of
\[
\frac{z^{(0)T}Az^{(0)}}{z^{(0)T}z^{(0)}}=\frac{\left(  z^{(0)},Az^{(0)}
\right)  }{\left(  z^{(0)},z^{(0)}\right)  }
\]

In the non-degenerate case,
\[
\max\frac{\left(  z^{(0)},Az^{(0)}\right)  }{\left(  z^{(0)},z^{(0)}\right)
}=\frac{\left(  x,A_{1}x\right)  }{\left(  x,x\right)  }=\lambda_{1}(A_{1})
\]
Hence,
\begin{align*}
\lambda^{(0)}  &  =\lambda_{1}(A_{1})\\
\left(  z^{(0)}\right)  ^{T}  &  =\left[  x^{T},\mathbf{0}^{T}\right]
\end{align*}
where the first $N$ elements of $z^{(0)}$ are $x$ and the rest $N$ elements
are all zeros. Let us look at the first order correction. Imposing the
identity for the first order expansion in (\ref{eigenequa}) gives
\begin{equation}
Az^{(1)}+Bz^{(0)}=\lambda^{(0)}z^{(1)}+\lambda^{(1)}z^{(0)}
\label{1ordercorrection}%
\end{equation}
Furthermore, we impose the normalization condition to $z$ (see (\ref{AAA})),
which leads to
\[
\left(  z^{(0)},z^{(1)}\right)  =0
\]
The first order correction to the principal eigenvector is orthogonal to the
zero order. Plugging this result in (\ref{1ordercorrection})
\begin{align*}
\left(  z^{(0)},Az^{(1)}+Bz^{(0)}\right)   &  =\lambda^{(0)}\left(
z^{(0)},z^{(1)}\right)  +\lambda^{(1)}\left(  z^{(0)},z^{(0)}\right) \\
\left(  A^{T}z^{(0)},z^{(1)}\right)  +\left(  z^{(0)},Bz^{(0)}\right)   &
=\lambda^{(1)}%
\end{align*}
that is
\begin{equation}
\left(  z^{(0)},Bz^{(0)}\right)  =\lambda^{(1)} \label{lamdafirstordernon}%
\end{equation}
Since $\left(  z^{(0)}\right)  ^{T}=\left(
\begin{array}
[c]{cc}%
x^{T} & \mathbf{0}^{T}%
\end{array}
\right)  $ and $B=\left[
\begin{array}
[c]{cc}%
\mathbf{0} & B_{12}\\
B_{12}^{T} & \mathbf{0}%
\end{array}
\right]  $, the first order correction is this non-degenerate case is null
$\lambda^{(1)}=0.$ Equation (\ref{1ordercorrection}) allows us to calculate
also the first order correction to the eigenvector
\begin{align*}
Az^{(1)}+Bz^{(0)}  &  =\lambda^{(0)}z^{(1)}\\
\left(  A-\lambda^{(0)}I\right)  z^{(1)}  &  =-Bz^{(0)}%
\end{align*}
$\left(  A-\lambda^{(0)}I\right)  $ is invertible out of its kernel $\left(
A-\lambda^{(0)}I\right)  z=\mathbf{0}$ (that is the linear space generated by
$z^{(0)}$) and since $Bz^{(0)}\perp z^{(0)}$ we have
\begin{equation}
z^{(1)}=\left(  \lambda^{(0)}I-A\right)  ^{-1}Bz^{(0)}
\label{xfirstordercorrection}%
\end{equation}
Let us look for the second order correction. Imposing the identification of
the second order term of (\ref{eigenequa}) we obtain
\[
Az^{(2)}+Bz^{(1)}=\lambda^{(0)}z^{(2)}+\lambda^{(1)}z^{(1)}+\lambda
^{(2)}z^{(0)}
\]
Projecting this vectorial equation on $z^{(0)}$ provides the second order
correction to $\lambda$
\begin{align*}
(z^{(0)},Az^{(2)}+Bz^{(1)})  &  =\lambda^{(0)}(z^{(0)},z^{(2)})+\lambda
^{(1)}\left(  z^{(0)},z^{(1)}\right)  +\lambda^{(2)}\left(  z^{(0)}
,z^{(0)}\right) \\
\lambda^{(2)}  &  =(z^{(0)},Az^{(2)})+(z^{(0)},Bz^{(1)})-\lambda^{(0)}
(z^{(0)},z^{(2)})\\
&  =\lambda^{(0)}(z^{(0)},z^{(2)})+(z^{(0)},Bz^{(1)})-\lambda^{(0)}
(z^{(0)},z^{(2)})\\
&  =(z^{(0)},Bz^{(1)})
\end{align*}
Substituting (\ref{xfirstordercorrection}) gives
\begin{equation}
\lambda^{(2)}=\left(  z^{(0)},B\left(  \lambda^{(0)}I-A\right)  ^{-1}
Bz^{(0)}\right)  \label{secondorder_A}%
\end{equation}
which can be further expressed as a function of the largest
eigenvalue/eigenvector of individual network $A_{1},A_{2}$ or their
interconnections $B_{12}$. Since
\[
Bz^{(0)}=\left(
\begin{array}
[c]{ll}%
0 & B_{12}\\
B_{12}^{T} & 0
\end{array}
\right)  \left(
\begin{array}
[c]{l}%
x\\
0
\end{array}
\right)  =\left(
\begin{array}
[c]{l}%
0\\
B_{12}^{T}x
\end{array}
\right)
\]
we have
\begin{align*}
\lambda^{(2)}  &  =\left(  B^{T}z^{(0)},\left(  \lambda^{(0)}I-A\right)
^{-1}Bz^{(0)}\right)  =\left(
\begin{array}
[c]{ll}%
0 & B_{12}x^{(0)}%
\end{array}
\right)  \left(
\begin{array}
[c]{ll}%
(\lambda^{(0)}I-A_{1}) & 0\\
0 & (\lambda^{(0)}I-A_{2})
\end{array}
\right)  ^{-1}\left(
\begin{array}
[c]{l}%
0\\
B_{12}^{T}x^{(0)}%
\end{array}
\right) \\
&  =\left(
\begin{array}
[c]{ll}%
0 & (x^{(0)})^{T}B_{12}%
\end{array}
\right)  \left(
\begin{array}
[c]{ll}%
(\lambda^{(0)}I-A_{1}) & 0\\
0 & (\lambda^{(0)}I-A_{2})
\end{array}
\right)  ^{-1}\left(
\begin{array}
[c]{l}%
0\\
B_{12}^{T}x^{(0)}%
\end{array}
\right) \\
&  =\left(
\begin{array}
[c]{ll}%
0 & (x^{(0)})^{T}B_{12}%
\end{array}
\right)  \left(
\begin{array}
[c]{ll}%
(\lambda^{(0)}I-A_{1})^{-1} & 0\\
0 & (\lambda^{(0)}I-A_{2})^{-1}%
\end{array}
\right)  \left(
\begin{array}
[c]{l}%
0\\
B_{12}^{T}x^{(0)}%
\end{array}
\right) \\
&  =(x^{(0)})^{T}B_{12}\left(  \lambda^{(0)}I-A_{2}\right)  ^{-1}B_{12}
^{T}x^{(0)}%
\end{align*}
which finishes the proof.

\subsection{Proof of Theorem \ref{theorem_perturbation_degenerate}}

In this case, the solution $z^{(0)}$ of the zero order expansion equation
\[
Az^{(0)}=\lambda^{(0)}z^{(0)}
\]
can be any combination of the largest eigenvector of the two individual
networks $x$ and $y$:
\begin{align*}
z^{(0)}  &  =c_{1}x+c_{2}y\\
c_{1}^{2}+c_{2}^{2}  &  =1
\end{align*}
and $\lambda^{(0)}=\lambda_{1}(A_{1})=\lambda_{1}(A_{2})$. The first order
correction of the largest eigenvalue correction in the non-degenerate case
(\ref{lamdafirstordernon}) holds as well for the generate case
\[
\left(  z^{(0)},Bz^{(0)}\right)  =\lambda^{(1)}
\]
which is however non-zero in the degenerate case due to the structure of
$z^{(0)}$ and is maximized by the right choice of $c_{1}$ and $c_{2}$. Thus,
\begin{align*}
\lambda_{1}(W)  &  =\max_{c_{1},c_{2}}\left(  \lambda_{1}(A_{1})+\alpha\left(
z^{(0)},Bz^{(0)}\right)  \right)  +o(\alpha^{2})\\
&  =\lambda_{1}(A_{1})+\max_{c_{1},c_{2}}\alpha c_{1}c_{2}\left(  \left(
B_{12}y,x\right)  +\left(  B_{12}^{T}x,y\right)  \right)  +o(\alpha^{2})\\
&  =\lambda_{1}(A_{1})+\frac{1}{2}\alpha\left(  \left(  B_{12}y,x\right)
+\left(  B_{12}^{T}x,y\right)  \right)  +o(\alpha^{2})\\
&  =\lambda_{1}(A_{1})+\frac{1}{2}\alpha\left(  x,B_{12}y\right)
+o(\alpha^{2})
\end{align*}
where $c_{1}c_{2}$ is maximum when $c_{1}=c_{2}=1/\sqrt{2}.$

One may also evaluate the second order correction to the largest eigenvalues
of the degenerate case. The following results we derived in the non-degenerate
case hold as well for the degenerate case
\[
\left\{
\begin{array}
[c]{c}%
\lambda^{(2)}=\left(  z^{(0)},Bz^{(1)}\right) \\
Az^{(1)}+Bz^{(0)}=\lambda^{(0)}z^{(1)}+\lambda^{(1)}z^{(0)}%
\end{array}
\right.
\]
The latter equation allows to calculate the first order correction to the
dominate eigenvector $z^{(1)}$:
\begin{equation}
(\lambda^{(0)}I-A)z^{(1)}=\left(  B-\lambda^{(1)}\right)  z^{(0)}
\label{th12egen1}%
\end{equation}
Any linear equation admits solutions when the constant term $\left(
B-\lambda^{(1)}\right)  z^{(0)}$ is orthogonal to the kernel of the adjoint
matrix of $\lambda^{(0)}I-A$.
\[
Ker(\lambda^{(0)}I-A)=\left\{  v:\ (\lambda^{(0)}I-A)v=0\right\}
\]
Apart from pathological cases, each the two interactive nets have non
degenerate maximum eigenvalues ($\lambda^{(0)}$) corresponding to the dominant
eigenvectors $x^{(0)}$ and $y^{(0)}$:
\[
\left\{
\begin{array}
[c]{lll}%
A_{1}x^{(0)} & = & \lambda^{(0)}x^{(0)}\\
A_{2}y^{(0)} & = & \lambda^{(0)}y^{(0)};
\end{array}
\right.
\]
in this case, the kernel of the matrix $\lambda^{(0)}I-A$ is just the linear
space generated by the maximum eigenvalue of the single nets:
\[
v=\left(
\begin{array}
[c]{l}%
ax^{(0)}\\
by^{(0)}%
\end{array}
\right)  .
\]
As stated, the constant term is orthogonal to the entire kernel:
\[
v\cdot\left(  B-\lambda^{(1)}\right)  z^{(0)}=\left(
\begin{array}
[c]{ll}%
ax^{(0)} & by^{(0)}%
\end{array}
\right)  \cdot\left(  B-\lambda^{(1)}\right)  z^{(0)}=a\left(  \left(
B_{12}^{T}x,y\right)  -\lambda^{(1)}\right)  +b\left(  \left(  x,B_{12}
y\right)  -\lambda^{(1)}\right)  =0.
\]
Therefore the solution of the previous equation exists and all solutions
differ by a vector in $Ker(\lambda^{(0)}I-A)$. It is worth stressing that the
value of $\lambda^{(2)}$ does not depend on this extra terms, that is
$\lambda^{(2)}$ is invariant under the transformation:
\[
\left\{
\begin{array}
[c]{l}%
x^{(1)}\rightarrow x^{(1)}+ax^{(0)}\\
y^{(1)}\rightarrow y^{(1)}+by^{(0)};
\end{array}
\right.  ;
\]
providing the normalization condition is respected:
\[
(z^{(0)})^{T}z^{(1)}=0
\]
that is
\[
(x^{(0)})^{T}x^{(1)}+(y^{(0)})^{T}y^{(1)}=0\rightarrow(x^{(0)})^{T}
(x^{(1)}+ax^{(0)})+(y^{(0)})^{T}(y^{(1)}+by^{(0)})=0
\]
that leads to $a=b$. Let us apply the transformation to $\lambda^{(2)}$:
\[
\lambda^{(2)}= \frac{1}{\sqrt{2}}\left(
\begin{array}
[c]{l}%
B_{12} y^{(0)}\\
B_{12}^{T} x^{(0)}%
\end{array}
\right)  ^{T} \left(
\begin{array}
[c]{l}%
x^{(1)}\\
y^{(1)}%
\end{array}
\right)  \to\hat{\lambda}^{(2)}= \frac{1}{\sqrt{2}}\left(
\begin{array}
[c]{l}%
B_{12} y^{(0)}\\
B_{12}^{T} x^{(0)}%
\end{array}
\right)  ^{T} \left(
\begin{array}
[c]{l}%
x^{(1)}+ax^{(0)}\\
y^{(1)}-ay^{(0)}%
\end{array}
\right)
\]
\[
\hat{\lambda}^{(2)}=\lambda^{(2)}+a\left[  (y^{(0)})^{T} B_{12}^{T}
x^{(0)}-(x^{(0)})^{T}B_{12}y^{(0)}\right]  =\lambda^{(2)}
\]
Therefore we are allowed to select a definite solution as we where fixing a
gauge. We will impose the orthogonality of $x^{(1)}$ with $x^{(0)}$ and
$y^{(1)}$ with $y^{(0)}$.

The linear operator $\lambda^{(0)}I-A_{1}+x^{(0)}(x^{(0)})^{T})$ and
$\lambda^{(0)}I-A_{1}$ operate identically over all vectors orthogonal to
$x^{(0)}$ (as all constant terms are in our case), while it behaves as an
identity in the linear space generated by $x^{(0)}$. Therefore to fix the
gauge one may substitute $\lambda^{(0)}I-A_{1}$ with $\lambda^{(0)}
I-A_{1}+x^{(0)}(x^{(0)})^{T})$. The same argument holds for $A_{2}$. This
allows us to provide $\lambda^{(2)}$ an explicit expression and to calculate
it algebraically. Equation~\ref{th12egen1} in components reads:
\[
\left\{
\begin{array}
[c]{lll}%
(\lambda^{(0)}I-A_{1} )x^{(1)} & = & B_{12}y^{(0)}-\lambda^{(1)}x^{(0)}\\
(\lambda^{(0)}I-A_{2} )y^{(1)} & = & B^{T}_{12}x^{(0)}-\lambda^{(1)}y^{(0)};
\end{array}
\right.
\]
that, after fixing the gauge, becomes:
\[
\left\{
\begin{array}
[c]{lll}%
(\lambda^{(0)}I-A_{1}+ x^{(0)}(x^{(0)})^{T})x^{(1)} & = & B_{12}
y^{(0)}-\lambda^{(1)}x^{(0)}\\
(\lambda^{(0)}I-A_{2}+ y^{(0)}(y^{(0)})^{T})y^{(1)} & = & B^{T}_{12}
x^{(0)}-\lambda^{(1)}y^{(0)};
\end{array}
\right.
\]
and hence the first-order correction to the dominant eigenvector can be
\[
\left\{
\begin{array}
[c]{lll}%
x^{(1)} & = & (\lambda^{(0)}I-A_{1}+ x^{(0)}(x^{(0)})^{T})^{-1} (
B_{12}y^{(0)}-\lambda^{(1)}x^{(0)} )\\
y^{(1)} & = & (\lambda^{(0)}I-A_{2}+ y^{(0)}(y^{(0)})^{T})^{-1} (B^{T}
_{12}x^{(0)}-\lambda^{(1)}y^{(0)});
\end{array}
\right.
\]

The second order correction can be finally be calculate algebraically,
resorting to the spectral properties of the isolated networks,
\[
\lambda^{(2)}=\frac{1}{2}\left(
\begin{array}
[c]{l}%
B_{12}y^{(0)}\\
B_{12}^{T}x^{(0)}%
\end{array}
\right)  ^{T}\left(
\begin{array}
[c]{ll}%
(\lambda^{(0)}I-A_{1}+x^{(0)}(x^{(0)})^{T})^{-1} & 0\\
0 & (\lambda^{(0)}I-A_{2}+y^{(0)}(y^{(0)})^{T})^{-1}%
\end{array}
\right)  \left(
\begin{array}
[c]{l}%
B_{12}y^{(0)}-\lambda^{(1)}x^{(0)}\\
B_{12}^{T}x^{(0)}-\lambda^{(1)}y^{(0)}%
\end{array}
\right)
\]
that is:
\begin{equation}
\lambda^{(2)}=(y^{(0)})^{T}B_{12}^{T}(\lambda^{(0)}I-A_{1}+x^{(0)}%
(x^{(0)})^{T})^{-1}(B_{12}y^{(0)}-\lambda^{(1)}x^{(0)})+(x^{(0)})^{T}%
B_{12}(\lambda^{(0)}I-A_{2}+y^{(0)}(y^{(0)})^{T})^{-1}(B_{12}^{T}%
x^{(0)}-\lambda^{(1)}y^{(0)}).
\end{equation}

\end{document}